\documentclass{article}

\usepackage{arxiv}
\usepackage[OT2, T1]{fontenc}
\usepackage[russian, english]{babel}
\usepackage{caption}
\usepackage{lscape} 
\usepackage{amsmath}
\usepackage{geometry}
\usepackage{amsthm}
\usepackage{adjustbox}
\usepackage{amsfonts}
\usepackage{mathtools}
\usepackage{verbatim}
\usepackage{algpseudocode,algorithm,algorithmicx}
\usepackage{mathtools}
\usepackage{comment}
\usepackage{verbatim, times, booktabs, dcolumn, setspace, fullpage}

\newtheorem{theorem}{Theorem}

\newtheorem{rem}{Remark}
\usepackage{graphicx}
\usepackage{wrapfig}
\usepackage{caption}
\usepackage{mathtools}
\usepackage{float}

\usepackage[english,noautotitles-r]{SASnRdisplay} 
\usepackage{url}
\usepackage{subfigure}
\lstdefinestyle{r-output}{
style = r-style,
style = r-output-user,
}

\usepackage{url}            
\usepackage{booktabs}       
\usepackage{amsfonts}       
\usepackage{nicefrac}       
\usepackage{microtype}      
\usepackage{graphicx}
\usepackage[numbers]{natbib}
\usepackage{doi}

\newcommand{\E}{\mathbf{E}}

\newcommand{\Probb}{\mathbf{P}}

\title{To impute or not to? Testing multivariate normality on incomplete dataset: Revisiting the BHEP test}

\author{ \href{https://orcid.org/0000-0002-0460-400X}{\includegraphics[scale=0.06]{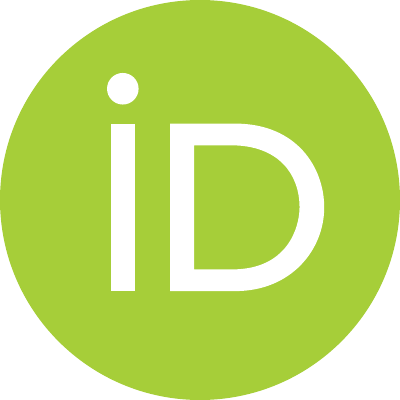}\hspace{1mm} Danijel G.    Aleksi\' c} \\
	University of Belgrade\\
	Faculty of Organizational Sciences, Faculty of Mathematics\\
	Belgrade, 11000, Serbia \\
	\texttt{danijel.aleksic@fon.bg.ac.rs} \\
	\And
	\href{https://orcid.org/0000-0001-8243-9794}{\includegraphics[scale=0.06]{orcid.pdf}\hspace{1mm}Bojana Milo\v sevi\' c} \\
	University of Belgrade\\
        Faculty of Mathematics\\
	Belgrade, 11000, Serbia \\
	\texttt{bojana.milosevic@matf.bg.ac.rs} \\
}
\date{}

\begin{document}
\maketitle

\begin{abstract}
    In this paper, we focus on testing multivariate normality using the BHEP test with data that are missing completely at random. Our objective is twofold: first, to gain insight into the asymptotic behavior of BHEP test statistics under two widely used approaches for handling missing data, namely complete-case analysis and imputation, and second, to compare the power performance of test statistic under these approaches.
   It is observed that under the imputation approach, the affine invariance of test statistics is not preserved. To address this issue, we propose an appropriate bootstrap algorithm for approximating p-values. Extensive simulation studies demonstrate that both mean and median approaches exhibit greater power compared to testing with complete-case analysis, and open some questions for further research. 
\end{abstract}

\section{Introduction}

Addressing missing data is very important in statistical analysis, as the presence of missing data directly impacts the validity and reliability of the conclusions drawn from the data. Failure to account for missing data can introduce bias, distort results, and undermine the accuracy of statistical methods. This problem has been studied from various perspectives. However, in the context of goodness-of-fit testing, the literature is very sparse.

On the other hand, testing multivariate normality is a crucial aspect of statistics as it provides a foundation for many statistical techniques and assumptions. In multivariate data analysis, researchers often assume that the data follows a multivariate normal distribution. Deviations from multivariate normality (MVN) can impact the validity of various statistical methods, such as multivariate analysis of variance, some properties of coefficients of linear regression, etc. Identifying departures from MVN allows researchers to make informed decisions about the appropriateness of chosen statistical methods and can guide the selection of alternative techniques if necessary. Ensuring MVN is also  important in fields like finance, biology, and social sciences, where accurate modeling of data distributions is essential for having valid inference and 
quality decisions based on statistical analyses. 
For this reason, many tests for testing MVN have been proposed so far. A thorough overview of them can be found in \cite{EbnerHenze2020}, while for some more recent results we refer to \cite{ebner2022testing}, \cite{gonzalez2022shapiro}, \cite{ejsmont2023test}.

Only a handful of tests were proposed for testing MVN on an incomplete sample, and they are based on the skewness and kurtosis estimates. Yamada et al. proposed, in \cite{yamada2015kurtosis}, a test for MVN based on a generalization of Mardia's statistic for measuring kurtosis. The test was suitable for two-step monotone incomplete data. Another such test is given by Kurita et al. in \cite{kurita2022multivariate}. Tan et al. in \cite{tan2005testing} proposed a test for multiple samples of possibly small sizes, that was, again, based on estimating kurtosis and skewness, utilizing multiple imputation and Gibbs sampler. 
To the best of our knowledge, MVN testing was not studied in the context of $L^2$-weighted test statistics, such as the statistic of the BHEP test. We aim to fill that gap.



The paper is structured as follows.  In Section \ref{sec:prerequisites} we make a brief review of BHEP test and its properties that are essential for studying the test in the presence of missing data, which is one of our aims. Then, in Section \ref{sec:miss} we consider behavior of test statistics under the complete-case approach. Some insights into tests' behavior under the imputation approach are also given in this section. To overcome the problem of dependence on the unknown parameters of the test statistic calculated on the imputed dataset, a bootstrap algorithm is proposed. Section \ref{sec:simulation}  contains results of an extensive simulation study that was conducted to compare Type I error preservation and power behavior of the test statistic under the complete-case approach, as well as some of the most widely used imputation methods, utilizing the proposed bootstrap algorithm.

Throughout the paper, all limits are taken when $n\to \infty$, where $n$ denotes the sample size; $\overset{D}{\to}$ denotes the convergence in distribution. In addition, $I_d$ denotes the $d\times d$ identity matrix. 

\section{BHEP: the prerequisites}\label{sec:prerequisites}

Here we focus on BHEP test which is one of the most well-known tests for testing MVN assumption and review its properties in the absence of missing data. For more details, we refer to  \cite{BaringhausHenze1988} and to \cite{HenzeWagner1997} for the generalization of the test.
All of the results stated here are known and will be utilized afterwards. 

Let $X_1, X_2, \dots, X_n$ be a sample of $n$ independent and identically distributed (IID) $d$-dimensional random column vectors. We are interested in testing the assumption that $X_1$ has some $d$-dimensional non-singular normal distribution, i.e. the hypothesis
\begin{align*} 
H_d \; : \; \text{the law of }X_1 \text{ is } \mathcal{N}_d (\mu, \Sigma), 
\end{align*}
for some $\mu \in \mathbb{R}^d$ and some non-singular covariance matrix $\Sigma$, where $\mathcal{N}_d (\mu, \Sigma)$ denotes the $d$-dimensional normal distribution with mean $\mu$ and covariance matrix $\Sigma$. 

Let 
\begin{align*} 
S_n = \frac{1}{n} \sum_{j=1}^n (X_j - \Bar{X}_n)(X_j - \Bar{X}_n)' 
\end{align*}
be a sample covariance matrix of a given sample, and let
\begin{align*} 
Y_j = S_n^{-\frac{1}{2}} (X_j - \Bar{X}_n), \quad j = 1, 2, \dots, n.
\end{align*}
If we denote by
\[ \psi_n(t) = \frac{1}{n} \sum_{j=1}^n \exp \left( it'Y_j \right), \quad t \in \mathbb{R}^d  \]
the empirical characteristic function of $Y_1, \dots, Y_n$, the test statistic of a BHEP test is given as
\begin{align}\label{Tn}
T_n = n \int_{\mathbb{R}^d} \Big| \psi_n(t) - \exp\big(-\frac12 \|t\|^2 \big)  \Big|^2 \varphi (t) \, \mathrm{d} t,
\end{align}
where $\varphi (t) = (2\pi)^{-d/2} \exp (- \frac12 |t|^2)$.

Baringhaus and Henze were successful in proving that $T_n$ can be represented as a weakly-degenerate $V$-statistic with estimated parameters. More precisely, they have shown that
\begin{align*}
T_n = n V_n (\lambda_n),
\end{align*}
where $\lambda = (\mu, \Sigma)$, $\lambda_n = (\Bar{X}_n, S_n)$, $V_n (\lambda) = n^{-2} \sum_{j,k=1}^n h(X_j, X_k; \lambda)$ and the closed-form expression for the kernel $h$ is given as
\begin{align}
    h(x_1, x_2; \lambda) &= \exp \Big( -\frac12 (x_1 - x_2)' \Sigma^{-1} (x_1 - x_2) \Big) \nonumber \\
    &- 2^{-d/2} \exp \Big( -\frac14 (x_1 - \mu)' \Sigma^{-1} (x_1 - \mu) \Big) \nonumber \\
    & - 2^{-d/2} \exp \Big( -\frac14 (x_2 - \mu)' \Sigma^{-1} (x_2 - \mu) \Big) + 3^{-d/2}, \label{h}
\end{align}
or, as they have shown
\begin{align}\label{h_preko_g}
    h(x_1, x_2; \lambda) = \int_{\mathbb{R}^d} g(x_1, t; \lambda) g(x_2, t; \lambda) \varphi (t) \, \mathrm{d} t,
\end{align}
where
\begin{align*}
    g(x,t; \lambda) = \cos \big(  t' \Sigma^{-1/2} (x - \mu) \big) + \sin  \big(  t' \Sigma^{-1/2} (x - \mu) \big) - \exp \big( -\frac12 \| t\|^2 \big).
\end{align*}
The authors then used existing theory about weakly-degenerate $U$- and $V$-statistics with estimated parameters that was developed by De Wet and Randles in \cite{deWetRandles1987} to derive the asymptotic distribution of the $T_n$ and to prove some further properties of the test. Specifically, the authors found that 
\begin{align}\label{convergence_original}
T_n \overset{D}{\to} \sum_{k=1}^{+ \infty} \kappa_k \chi_{1,k}^2,
\end{align}
where $\chi_{1,k}^2$ are IID $\chi_1^2$-distributed random variables, and $\kappa_k$-s are eigenvalues of the operator $A$ defined as
\begin{align}\label{A_operator}
    Af(x) = \int_{\mathbb{R}^d} h_* (x,y) f(y) \varphi (y) \, \mathrm{d} y.
\end{align}
The closed form of the kernel $h_*$ is derived in \cite{BaringhausHenze1988}, while the problem of  deriving eigenvalues is recently studied in \cite{EbnerHenze2023}.




Although the asymptotic distribution \eqref{convergence_original} can approximated, in practice, the null distribution is usually approximated using Monte Carlo simulation.
 Since the $T_n$ is affine invariant, one can safely assume that the null distribution is $d$-variate standard normal, and simulate the null distribution for the desired sample size. The empirical quantiles are then used to construct the rejection region of the test.

\section{The BHEP test: challenges of incomplete dataset}\label{sec:miss}

Now assume that we have the sample $X_1, \dots, X_n$ of IID $d$-dimensional random column vectors and that some of data are missing. Let us introduce,  for every $j = 1,2, \dots, n$ and $k = 1,2, \dots, d$, the \textit{response indicator} for $X_{jk}$:
\begin{align*}
    R_{jk} = 
    \begin{cases}
        1, & \text{if } X_{jk} \text{ is observed}, \\
        0, & \text{otherwise},
    \end{cases}
\end{align*}
where $X_{jk}$ is $k$-th component of $X_j$. Denote $R_j = (R_{j1}, \dots, R_{jd})'$. We assume that $R_j$s are mutually independent (which does not need to follow from the independence of $X_j$s), and denote $q = (q_1, \dots, q_d)' = \E (R_1)'$. Finally, we use $\odot$ to denote the Hadamard-Schur componentwise multiplication of vectors. 

We say that data are missing completely at random (MCAR) if the distribution of the response indicator does not depend neither on the observed, nor on the missing data. Alternatively, when response indicator depends only on the observed, but not the missing data, we say that data are missing at random (MAR). Otherwise, data are missing not at random (MNAR). The definitions are due to Little and Rubin, and detailed explanations can be found in \cite{RubinLittle1987}.
Here we focus our attention on the case when the data are MCAR.

In what follows we consider an expanded sample
\begin{align*}
    (X_1', R_1')', \dots, (X_n', R_n')',
\end{align*}
which is suitable for expressing the test statistic in the presence of missing data.

\subsection{Complete-case approach }

The complete-case approach,i.e. when every observation that is not completely observed is removed from the sample, is very common in practice. It is mostly used when data are MCAR and the missingness rate is not very high, since most estimates, such as sample mean, remain unbiased and consistent. Intuitively, one can expect that a test statistic will preserve properties under MCAR and complete-case, since we are working with the "representative" subset of the original sample. The claim remains true for the test statistic of BHEP test. Although intuitive, it deserves rigorous proof. Furthermore, the proof of the similar result for non-degenerate $U$-statistics was anything but trivial, as seen in \cite{aleksic2023etAl}. 
\begin{theorem}\label{teorema_cc}
    Let the null hypothesis and the MCAR assumption hold, and let $T_n$ be as in \eqref{Tn}. Let $\hat{T}_n$ be the same statistic, but calculated on the completely observed sample units. Then, it holds that $nT_n$ and $\hat{n}\hat{T}_n$ have the same asymptotic distribution, where $\hat{n}$ is the number of complete cases.
\end{theorem}

\begin{proof}
    First of all, note that $q$ is not considered to be a distribution parameter here, but more as a nuisance parameter, that is some known number. The test statistic does not use its estimate in the sense of De Wet and Randles. 

    Secondly, note that $\hat{T}_n$ is also affine invariant with respect to transformations of $X_1, \dots, X_n$. It follows in the same manner as for the $T_n$. For more detail, one can consult Baringhaus and Henze (\cite{BaringhausHenze1988}, p. 3).
    
    Now, if we denote $W_j = \prod_{k=1}^d R_{jk}$ for every $j$, one can easily see that $\hat{n} = \sum_{j=1}^n W_j$, and that $\hat{n}/n \overset{P}{\to} q_{\Pi}$ under MCAR, where $q_\Pi$ denotes the product of the components of $q$. From this point onwards, we can consider that we work with the expanded sample
    \begin{align*}
    (X_1', W_1)', \dots, (X_n', W_n)'.
    \end{align*}
    If we introduce 
    \begin{align*}
        \hat{g} ((x, w), t; \lambda) = wg (x, t; \lambda) ,
    \end{align*}
    and 
    \begin{align*}
        \hat{h} ((x,w_x), (y, w_y) ; \lambda) &= \int_{\mathbb{R}^d} \hat{g} ((x, w_x), t; \lambda) \hat{g} ((y, w_y), t; \lambda) \varphi (t) \,\mathrm{d}t = w_x w_y h(x,y; \lambda),
    \end{align*}
    one can see that $\hat{\epsilon} (t; \lambda) := \E (\hat{g} ((X, W), t; \lambda)) = q_\Pi \E (g(X, t; \lambda)) = q_\Pi \epsilon (t; \lambda)$, where $\epsilon (t; \lambda)$ is as in \cite{BaringhausHenze1988}. Having these direct relations between kernels makes regularity conditions of De Wet and Randles follow trivially from those in \cite{BaringhausHenze1988}. The only one that requires attention is Condition (2.10), i.e. the expansion of parameter estimates, which we now present.
    
    The parameters here are estimated on the complete cases, being
    \begin{align*}
        \hat{\mu} = \frac{1}{\hat{n}} \sum_{j=1}^n W_j X_j = \frac{1}{n} \sum_{j=1}^n \frac{W_j}{q_\Pi}  X_j + o_P(1/\sqrt{n})
    \end{align*}
    and
    \begin{align*}
        \hat{\Sigma} = \frac{1}{\hat{n}} \sum_{j=1}^n W_j (X_j - \Bar{X}_j) (X_j - \Bar{X}_j)' = \frac{1}{n} \sum_{j=1}^n  \frac{W_j}{q_\Pi} (X_j - \Bar{X}_j) (X_j - \Bar{X}_j)' + o_P(1/\sqrt{n}).
    \end{align*}
    This allows us to use 
    \begin{align*}
        \hat{\alpha} (x,w) = \frac{w}{q_\Pi} \alpha (x)
    \end{align*}
    to express
    \begin{align*}
        (\hat{\mu}, \hat{\Sigma}) = (0, I_d) + \frac{1}{n} \sum_{j=1}^n \hat{\alpha} (X_j, W_j) + o_P(1/\sqrt{n}),
    \end{align*}
      where $\alpha(x)=(x,xx'-I_d)$ (see \cite{BaringhausHenze1988}).
    Having this, and the direct relation $\hat{\epsilon}_1 (t; 0, I_d) = q_\Pi \epsilon_1 (t; 0, I_d)$, where $\epsilon_1$ is as in \cite{BaringhausHenze1988},  we proceeed in the same manner as Baringhaus and Henze and obtain the
    kernel
    \begin{align}
        \hat{h}_* ((x,w_x), (y, w_y) ) &= \int_{\mathbb{R}^d} \left( \hat{g} ((x, w_x), t; 0, I_d) - \hat{\epsilon}_1 (t; 0, I_d) \hat{\alpha} (x, w_x)\right)\cdot \nonumber \\
        &\quad \quad \quad \quad \cdot \left( \hat{g} ((y, w_y), t; 0, I_d) - \hat{\epsilon}_1 (t; 0, I_d) \hat{\alpha} (y, w_y)\right) \varphi (t) \, \mathrm{d} t \nonumber \\ 
        &= w_x w_y h_*(x,y). \label{h_hat_star}
    \end{align}
    By De Wet and Randles, and similarly to Baringhaus and Henze, we have that
    \begin{align*}
        n \hat{T}_n \overset{D}{\to} \sum_{k=1}^{+ \infty} \zeta_k \chi_{1, k}^2,
    \end{align*}
    where $\zeta_k$ are the eigenvalues of an integral operator
    \begin{align*}
       Bg(x,w_x)  &= \sum_{w_y \in \{0,1\}} \left(\int_{\mathbb{R}^d} w_x w_y h_*(x,y) g(y, w_y) \varphi(y) \, \mathrm{d}y \right)\Probb \{ W = w_y \} \\
       &=  w_x q_{\Pi} \int_{\mathbb{R}^d} h_* (x,y) g(y, 1) \varphi (y) \, \mathrm{d} y.
    \end{align*}
    Now we see that $g(x,w_x)$ is an eigenfunction of $B$ if and only if $g(x,w_x)=w_xf(x)$, where $f(x)$ is an eigenfunction of $A$, where $A$ is from \eqref{A_operator}. Furthermore, every eigenvalue of $B$ is the eigenvalue of $A$ multiplied by $q_{\Pi}$. Then, it follows that
    \begin{align*}
        n\hat{T}_n \overset{D}{\to} q_\Pi \sum_{k=1}^{+ \infty} \kappa_k \chi_{1, k}^2, 
    \end{align*}
    where the right-hand side is as in \eqref{convergence_original}. Since $n/\hat{n} \overset{P}{\to} 1/q_\Pi$, Slytsky's theorem gives us that
    \begin{align*}
         \hat{n}\hat{T}_n \overset{D}{\to} \sum_{k=1}^{+ \infty} \kappa_k \chi_{1, k}^2, 
    \end{align*}
    which concludes the proof.
\end{proof}

\subsection{Imputation approach}

A very common approach when working with missing data is to impute the missing values using some imputation method. However, in practice, after the imputation, the analysis is conducted as if the data were complete. However, the imputation clearly changes the distribution of the data. Furthermore, in some papers, e.g. \cite{aleksic2023etAl}, the change in distribution that imputation produces has been exactly quantified. 

In this subsection, we discuss the BHEP test in the context of the imputed dataset, illustrating the methodology and problems that arise using the example of sample mean imputation. To be more formal, we replace every $X_{jk}$, $j=1,2, \dots, n$, $k = 1,2,\dots, d$, that is not observed, with adequate mean $\frac{1}{\sum_{j=1}^n R_{jk}} \sum_{j=1}^n X_{jk} R_{jk}$, calculated on available column units. Having results of De Wet and Randles (\cite{deWetRandles1987}), we first use components of the theoretical mean value $\mu = (\mu_1, \dots, \mu_d)'$ of the distribution as the imputed values, and then move to estimated ones.

By noting that instead of saying that data value is getting replaced with corresponding $\mu_i$ when missing, and remaining unchanged when observed, one can neatly write that every $X_j$ from the sample, $j = 1, \dots, n$, is being replaced with
\begin{align*}
    (X_j - \mu) \odot R_j + \mu.
\end{align*}
Now, the test statistic on the (sample) imputed dataset can be written as
\begin{align*}
    \Tilde{T}_n = n \Tilde{V}_n (\Tilde{\lambda}_n),
\end{align*}
where $\Tilde{\lambda}_n = (\Tilde{X}_n, \Tilde{S}_n)$ are parameters estimated from the incomplete data. The vector of means is estimated as mentioned above, i.e. $\Tilde{X}_n \frac{1}{\sum_{j=1}^n R_j} \sum_{j=1}^n X_j \odot R_j$,  where we use slight abuse of notation since the division is also componentwise, $R_j$ being a vector for every $j$. We choose to estimate the covariance matrix $\Sigma$ using only complete cases, i.e.
\begin{align}\label{sigma_tilda}
    \Tilde{\Sigma} = \frac{1}{\sum_{j=1}^n W_j} \sum_{j=1}^n W_j (X_j - \Bar{X}_n) (X_j - \Bar{X}_n)'.
\end{align}
The another choice would be to estimate $\Sigma$ as the sample covariance matrix calculated on the imputed dataset, but for that estimate it is difficult to verify the condition (2.10) of De Wet and Randles, whose results one may aim to utilize in this context. Furthermore, such estimate in many cases, e.g. mean imputation which we will study, is not consistent. 

Going further, one can see that the auxiliary $V$-statistic $\Tilde{V}_n (\lambda)$ 
can be written as $\Tilde{V}_n (\lambda) = n^{-2} \sum_{j,k=1}^n \Tilde{h}(X_j, X_k; \lambda)$, where
\begin{align}\label{h_tilda}
    \Tilde{h} \left((x_1, r_1), (x_2, r_2); \lambda \right) = h \left( (x_1 - \mu) \odot r_1 + \mu, (x_2 - \mu) \odot r_2 + \mu ; \lambda \right)
\end{align}
and, similarly
\begin{align*}
    \Tilde{g} \left( (x,r), t; \lambda \right) = g \left( (x - \mu) \odot r + \mu , t ; \lambda \right).
\end{align*}
Since integrations depend only on $t$, the relation \eqref{h_preko_g} holds for $\Tilde{h}$ and $\Tilde{g}$. To rely on the known results for asymptotics of weakly-degenerate \textit{V}-statistics with estimated parameters, one needs to verify that conditions of De Wet and Randles (\cite{deWetRandles1987}) hold. 

It is readily seen that
\begin{align}
    \Tilde{g} &\big( (x,r), t ; \lambda \big) \nonumber \\
    &= \cos \big(  t' \Sigma^{-1/2} ((x - \mu) \odot r)\big) + \sin \big(  t' \Sigma^{-1/2} ((x - \mu) \odot r)\big) - \exp \big( - \frac{1}{2} \| t \|^2  \big)  \nonumber \\
    &= \cos \big(  t' \Sigma^{-1/2} (x \odot r)\big) \cos \big( t' \Sigma^{-1/2} (\mu \odot r) \big) + \sin \big(  t' \Sigma^{-1/2} (x \odot r)\big) \sin \big( t' \Sigma^{-1/2} (\mu \odot r) \big) \nonumber \\
    &\quad \quad \quad + \sin \big(  t' \Sigma^{-1/2} (x \odot r)\big) \cos \big( t' \Sigma^{-1/2} (\mu \odot r) \big) - \cos \big(  t' \Sigma^{-1/2} (x \odot r)\big) \sin \big( t' \Sigma^{-1/2} (\mu \odot r) \big) \nonumber \\
    & \quad \quad \quad - \exp \big( - \frac{1}{2} \| t \|^2  \big). \label{g_tilda_raspisano}
\end{align}
If we go back to the kernel $h$ as in \eqref{h} and substitute every $X_j$ with an (sample) estimate, an looking at $\Tilde{\Sigma}$ as in \eqref{sigma_tilda}, we can easily see that $\Tilde{V}_n(\Tilde{\lambda}_n)$ is invariant with respect to translations, i.e. does not depend on the expected value of the data. Unfortunately, it does depend on the covariance matrix of the data, which will be seen in more detail later on. From this point onwards, it will be assumed that data are IID sampled from the $d$-variate normal distribution $\mathcal{N}_d (0, \Delta)$, for some positive definite covariance matrix $\Delta$. 

\begin{rem}\label{remark_odot}
    Note that for the column vectors $a,b,c$ from $\mathbb{R}^d$ it holds that $a' (b \odot c) = (a' \odot c') b$.
\end{rem}

Following the idea from Baringhaus and Henze \cite{BaringhausHenze1988}, let us focus on the first term of the first summand of \eqref{g_tilda_raspisano}, that is $\cos \left(  t' \Sigma^{-1/2} (x \odot r)\right)$ (the other term is constant with respect to $x$). Assuming for a moment that $R$ is a constant vector equal $r$, one is able, relying on the Remark \ref{remark_odot}, to calculate that
\begin{align*}
    \E \Big(  \cos \big(  t' \Sigma^{-1/2} (X \odot R)\big) \, \Big| \, R = r \Big) &= \int_{\mathbb{R}^d} \cos \big(  t' \Sigma^{-1/2} (x \odot r)\big) f_{\mathcal{N}_d (0, \Delta)} (x) \, \mathrm{d} x \\
    &= \int_{\mathbb{R}^d} \cos \big(  \big( (t' \Sigma^{-1/2}) \odot r' \big) x\big) f_{\mathcal{N}_d (0, \Delta)} (x) \, \mathrm{d} x,
\end{align*}
which is exactly the characteristic function of the $\mathcal{N}_d (0, \Delta)$ distribution, calculated at the point $(t' \Sigma^{-1/2}) \odot r'$, so we obtain
\begin{align*}
    \E \Big(  \cos \big(  t' \Sigma^{-1/2} (X \odot R)\big) \, \Big| \, R = r \Big) &= \exp \Big(  -\frac12 \big( (t' \Sigma^{-1/2}) \odot r' \big) \Delta \big( (t' \Sigma^{-1/2}) \odot r' \big)' \Big).
\end{align*}
Conducting similar calculations on the other terms in \eqref{g_tilda_raspisano}, one can obtain that
\begin{align*}
    \E \Big( \Tilde{g} \big(  (X, R), t ; \lambda \big) \, \Big| \, R = r  \Big) &= \Big( \cos \big( t' \Sigma^{-1/2} (\mu \odot r) \big) - \sin \big( t' \Sigma^{-1/2} (\mu \odot r) \big) \Big)  \\
    &   \quad \cdot \exp \Big(  -\frac12 \big( (t' \Sigma^{-1/2}) \odot r' \big) \Delta \big( (t' \Sigma^{-1/2}) \odot r' \big)' \Big) - \exp \Big( -\frac12 \| t \|^2 \Big).
\end{align*}
The next step would be to obtain the expected value with respect to $R$ of the above expression, which is a finite sum over all possible values of $d$-tuples of zeros and ones. 

To be able to use conditions of De Wet and Randles, as we did proving Theorem \ref{teorema_cc}, one needs to verify conditions 2.9 - 2.11 from \cite{deWetRandles1987}. The key step is to calculate the expected value
\begin{align*}
    \Tilde{\epsilon} (t, \lambda) = \E \left( \Tilde{g} \left( (X,R), t ; \lambda \right)  \right),
\end{align*}
assuming the underlying  $\mathcal{N}_d (0, \Delta)$ distribution, and then find the vector of its partial derivatives, calculated at true value of parameters (here $\lambda = (0, \Delta)$), denoted by $\Tilde{\epsilon}_1 (t;  \lambda)$. This would then be used to find the kernel similar to \eqref{h_hat_star}, and to determine the asymptotic distribution using its eigenvalues.

However, further calculations would be of no great help, since this kernel, and subsequently its operator eigenvalues, depend on the unknown distribution parameters, and it is not very likely that one could be able to obtain them analytically. 
In the complete-data case, the null distribution can be simulated by sampling from the multivariate standard normal distribution, and empirical quantiles can be used for calculating the critical values. This is all due to original statistic being affine invariant, and consequently distribution-free. Our statistic is neither, so the reasonable approach is to use a bootstrap algorithm to approximate p-value. Here we suggest the usage of Algorithm \ref{bootstrap_algorithm}.

%


\begin{rem}
    We note that the Algorithm \ref{bootstrap_algorithm} is modifications of a bootstrap algorithm from \cite{jimenez2003bootstrapping}, that is designed to work with a complete sample. One of the goals of the empirical study that follows is to examine its properties under various methods of imputation.
\end{rem}

\begin{algorithm}
  \caption{A bootstrap algorithm for MCAR data}\label{bootstrap_algorithm}
  \begin{algorithmic}[1]
    \State Start with the incomplete sample $\mathbf{x}=(x_1, \dots, x_{n_1})$;
    \State Estimate covariance matrix $\Sigma$ by $\Tilde{\Sigma}$, calculated on the complete cases; estimate mean vector $\mu$ by $\Tilde{\mu}$ on the dataset that is imputed by a chosen method; 
    \State Calculate estimate $\Tilde{p}$ of the vector of by-column missingness probabilities using response indicator averages;
    \State Generate bootstrap sample $\mathbf{x^*}=(x^*_1, \dots, x^*_{n_1})$  from $\mathcal{N} (\Tilde{\mu}, \Tilde{\Sigma})$;
    \State Generate missingness in $\mathbf{x^*}$ according to MCAR and probabilities $\Tilde{p}$ and impute the sample using chosen method to obtain the imputed sample $\mathbf{x^*}_{\mathrm{imp}}$;
    \State Obtain the value $T_n^* (\mathbf{x^*}_{\mathrm{imp}})$ of the BHEP statistic on the imputed dataset;
    \State Repeat steps 4-5 $B$ times to obtain sequence of bootstrapped test statistics $T_{n, 1}^*, \dots, T_{n,B}^*$;
    \State Reject the null hypothesis for the significance level $\alpha$ if $T_n(\mathbf{x})$ is greater than the $(1 - \alpha)$-quantile of the empirical bootstrap distribution of $(T_{n, 1}^*, \dots, T_{n,B}^*)$.
  \end{algorithmic}
\end{algorithm}

\section{Power study}\label{sec:simulation}

In this section, an extensive simulation study is conducted to give an answer to a crucial question that motivated our work: Is it better to impute the data, or to use complete-case approach? We observe different scenarios, varying underlying data distribution, missingness percentage, as well as imputation methods.
Here we highlight the most significant simulation results that reinforce our key points. Additional results can be found in Appendix.

{
Empirical sizes and powers are obtained  using Monte Carlo procedure and bootstrap algorithm \ref{bootstrap_algorithm} with $N=2000$ replicates and $B=1000$ bootstrap cycles. In order to make a fair comparison we apply the same procedure along different approaches, although we are aware that the usage of bootstrap in the complete-case is not necessary. Everything is done for the level of significance $\alpha=0.05$.
}

The choice of imputation methods came down to mean imputation, median imputation, as well as 3- and 6-nearest neighbor imputation. For the first two, built-in functions from the \texttt{R} \texttt{missMethods} (\cite{rockel2022missmethods}) package are used, and for the $kNN$ we use \texttt{knn.impute} function from the package \texttt{bnstruct} (\cite{franzin2017bnstruct}). 


MCAR data are created using the \texttt{delete\_MCAR} function from the \texttt{missMethods} package. To examine properties of each method, with the increase of sample size,  sample sizes of 30, 60, 90 and 120 are considered. As alternatives to the null hypothesis, Student's $t$ distributions are used, with 5, 7 and 11 degrees of freedom. Also, data of several different dimensions are generated, and the covariance/scale matrix is also varied.

Before we present results of our study, we point to one most common misuse of BHEP test under the missing data, which consists of imputing the dataset, but proceeding with the data analysis procedure that was designed only for complete samples. From Figure \ref{fig:2d_stdnorm_oldfashion}, it is clear that, in the context of testing MVN, this is not a feasible approach, since the Type I error is severely distorted.

\begin{figure}[ht]
    \centering
    \includegraphics[width=\textwidth]{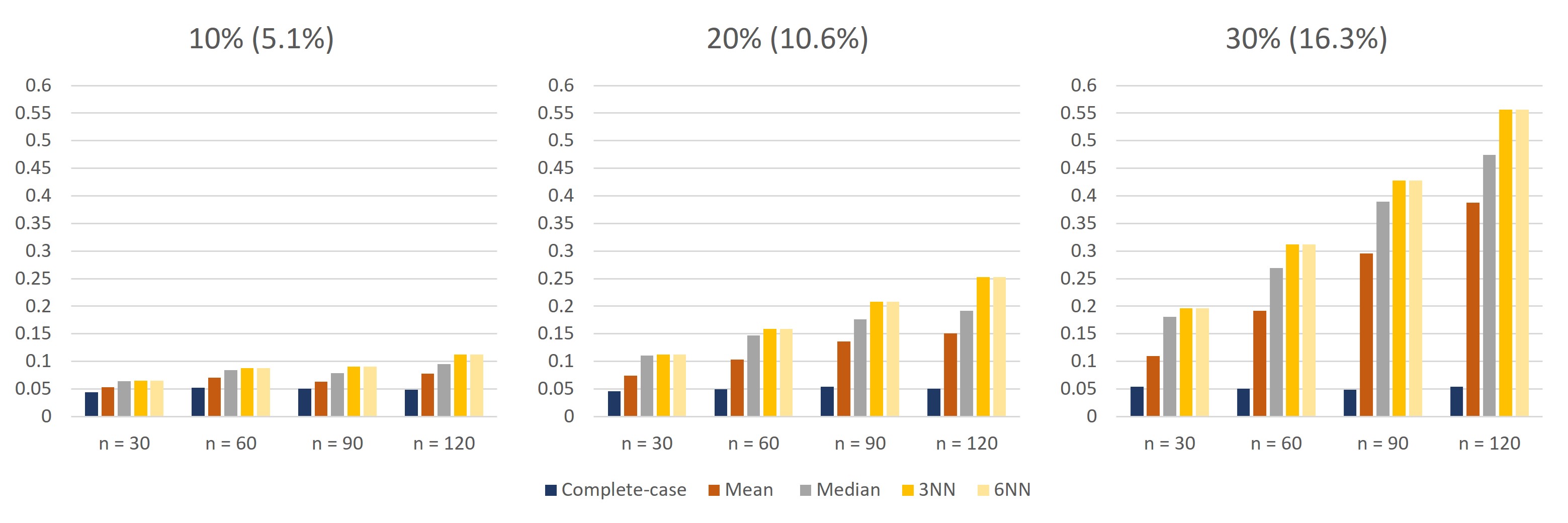}
    \caption{Empirical test sizes for underlying 2D standard normal distribution and MCAR data, ignoring that the data were imputed (First percentage = probability that a row is incomplete, second percentage = probability that a value is missing)}
    \label{fig:2d_stdnorm_oldfashion}
\end{figure}

\begin{figure}[ht]
    \centering
    \includegraphics[width=\textwidth]{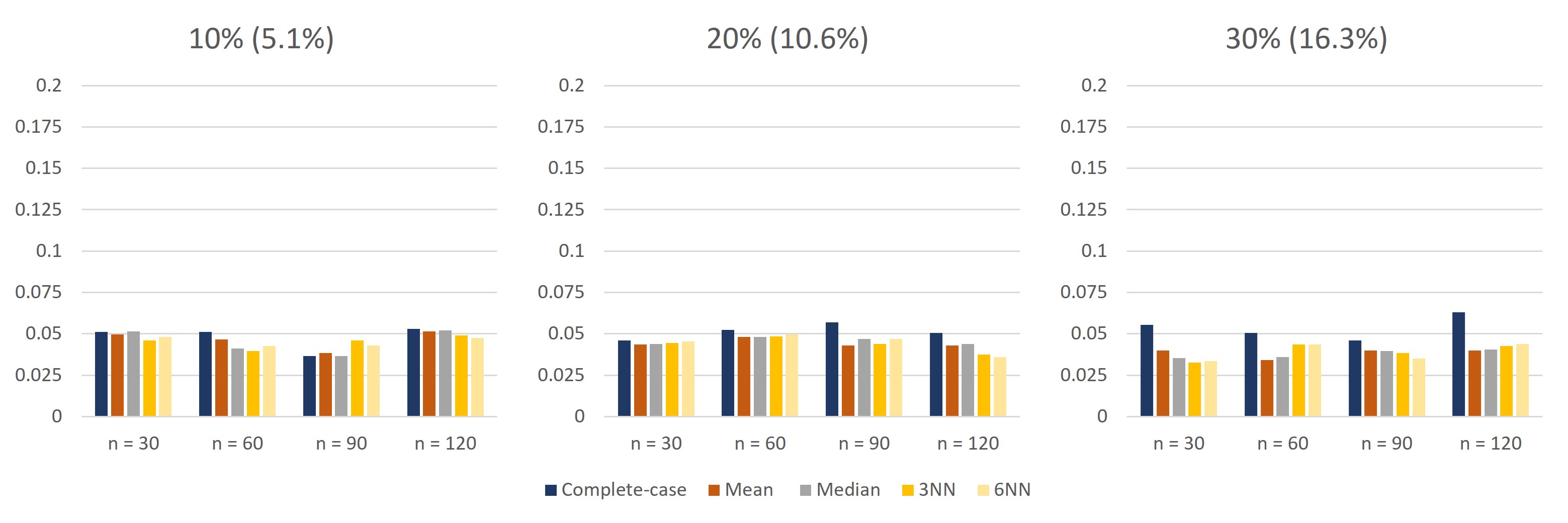}
    \caption{Empirical test sizes for underlying 2D standard normal distribution and MCAR data. }
    \label{fig:2d_stdnorm}
\end{figure}

As one can see from the Figure \ref{fig:2d_stdnorm}, for the MCAR data and bivariate standard normal distribution, complete-case approach presents itself as the best in terms of Type I error preservation, and is, especially for moderate missingness, followed closely by the other methods. For 3D case, however, as seen in the Figure \ref{fig:3d_stdnorm}, $kNN$ methods appear to be slightly more conservative and have empirical size further from the desired level $0.05$. This becomes emphasized as missingness probability starts to grow. However, in most of the real-world scenarios, where missingness is moderate, we can expect all of the methods to remain well-calibrated. Similar conclusions can be drawn for different correlation structures, as can be seen from the figures in the Appendix. 

\begin{figure}[ht]
    \centering
    \includegraphics[width=\textwidth]{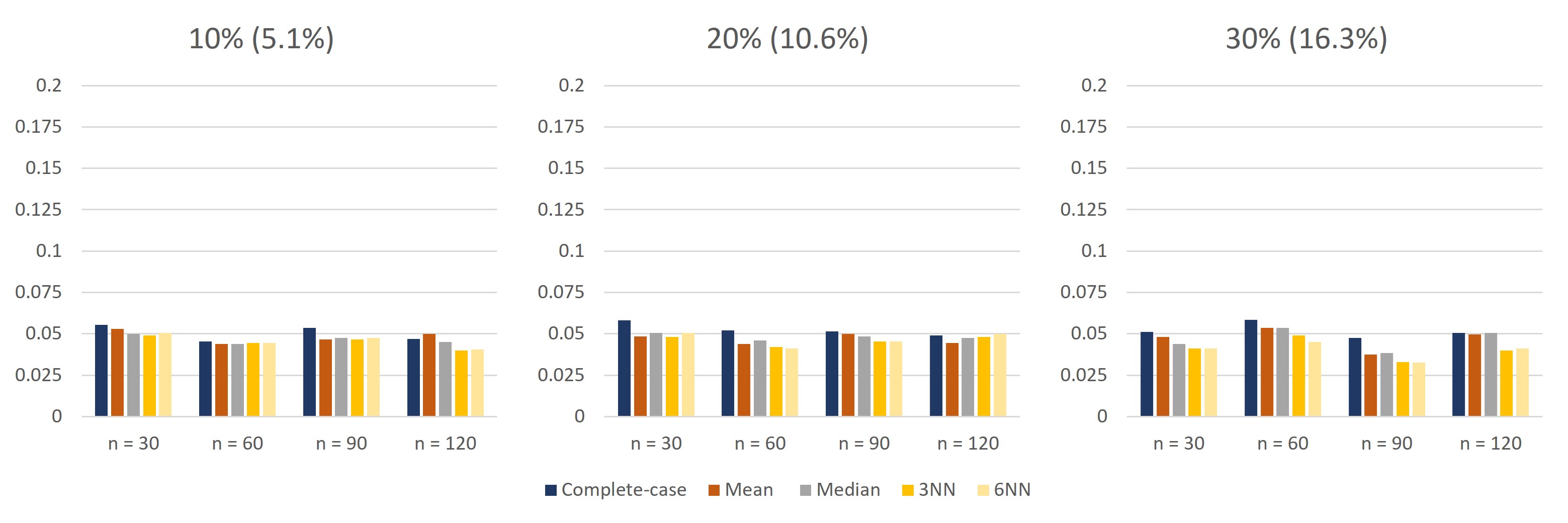}
    \caption{Empirical test sizes for underlying 3D standard normal distribution and MCAR data}
    \label{fig:3d_stdnorm}
\end{figure}

\begin{figure}[ht]
    \centering
    \includegraphics[width=\textwidth]{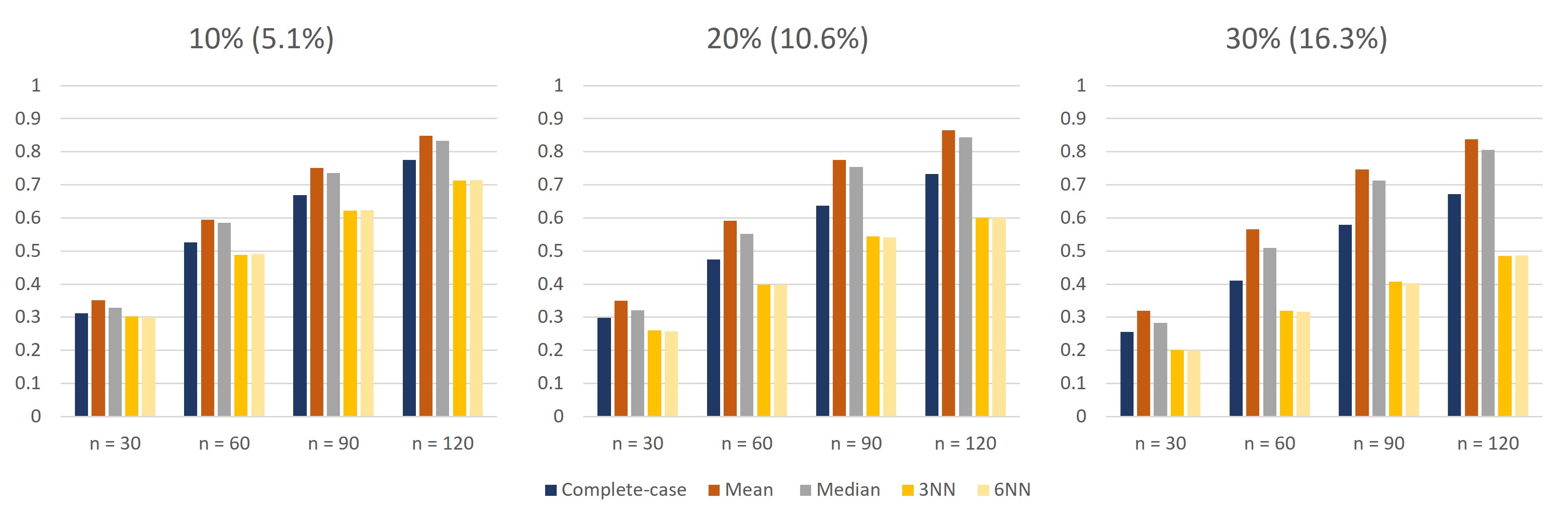}
    \caption{Powers for underlying 2D $t_5$ distribution and MCAR data}
    \label{fig:2d_t5}
\end{figure}

This being said, we shift our focus on the power comparison. As can be seen from the Figure \ref{fig:2d_t5}, mean value imputation provides the greatest empirical powers, followed closely by median imputation. The $kNN$ approaches significantly lack behind, while complete-case performs somewhere in the middle. One needs to point out that the advantage of mean imputation approach is even more significant, if we remember that Empirical test sizes have shown that complete-case approach have a higher tendency to reject the null hypothesis. The same relations uphold for all of the other $t$-distributions, both standard and scaled, as well for all of the observed dimensions. The results can be found in the Supplementary material.

To summarize, the best way to perform a normality testing is to use the Algorithm \ref{bootstrap_algorithm} with mean value imputation, in the case where one is able to conclude that their data are MCAR. This can be done using, for example, well-known test developed by Little (see \cite{Little1988}), or some of the recently developed tests by Bordino and Berrett (\cite{bordino2024tests}), Berrett and Samworth (\cite{berrett2023optimal}), or Aleksić (\cite{aleksic2023novel}).

\section{Concluding remarks and further outlook}

This paper has several main contributions. First, we have proved that complete-case analysis can be used for testing MVN under the MCAR data since the test statistic in that case has the same asymptotic distribution as the original one, calculated on the complete sample. We have provided insight into the limiting distribution of test statistics under the imputation procedure and demonstrated that in this scenario, the affine invariance property is not preserved. To overcome this, we have proposed a bootstrap algorithm for testing MVN that preserves the Type I error of the test.  
Next,  we have demonstrated the flaw of the commonly used practice of treating the imputed dataset as if it were complete and proceeding with an analysis ignoring the imputation.
Finally, we compare the powers of the test under complete-case and mean, median and $kNN$ imputation.

One of the natural extensions of this work is to extend the study to other commonly used multivariate tests for normality, and to go further by exploring the behavior of recent goodness-of-fit tests for other multivariate distributions (e.g. \cite{karling2023goodness}, \cite{ebner2024unified}).
Another direction would be to explore tests' properties under MAR settings. Preliminary results indicate that all considered approaches significantly influence the distribution of the tests and, consequently, the testing process. Developing a bootstrap algorithm that addresses this scenario remains an open question.





\section*{Disclosure statement}

The authors have no conflict of interests to disclose.

\section*{Funding}

The work of D. Aleksić is supported by the Ministry of Science, Technological Development and Innovations of the Republic of Serbia (the contract 451-03-66/2024-03/200151). The work of B. Milošević is supported by the Ministry of Science, Technological Development and Innovations of the Republic of Serbia (the contract 451-03-66/2024-03/200104), and by the COST action CA21163 - Text, functional and other high-dimensional data in econometrics: New models, methods, applications (HiTEc).

\section*{Notes on contributors}

D. Aleksić and B. Milošević have contributed equally to this work.

\appendix

\section{Additional simulation results}

Let 
\[ \Sigma_1 = \begin{bmatrix}
    1 & 0.5 \\
    0.5 & 1
\end{bmatrix}, \quad    \Sigma_2 = \begin{bmatrix}
    1 & 0.5 & 0.5 \\
    0.5 & 1 & 0.5 \\
    0.5 & 0.5 & 1
\end{bmatrix} \]



\begin{figure}[ht]
    \centering
    \includegraphics[width=\textwidth]{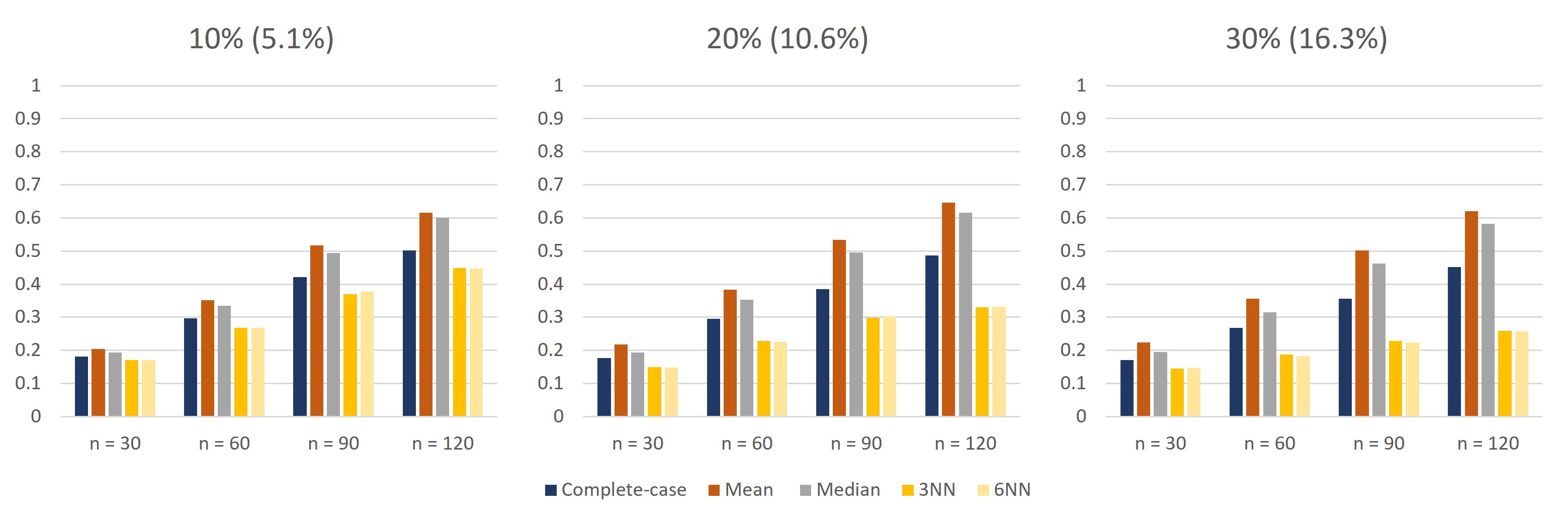}
    \caption{Powers for underlying 2D $t_7$ distribution and MCAR data}
    \label{fig:2d_t7}
\end{figure}

\begin{figure}[ht]
    \centering
    \includegraphics[width=\textwidth]{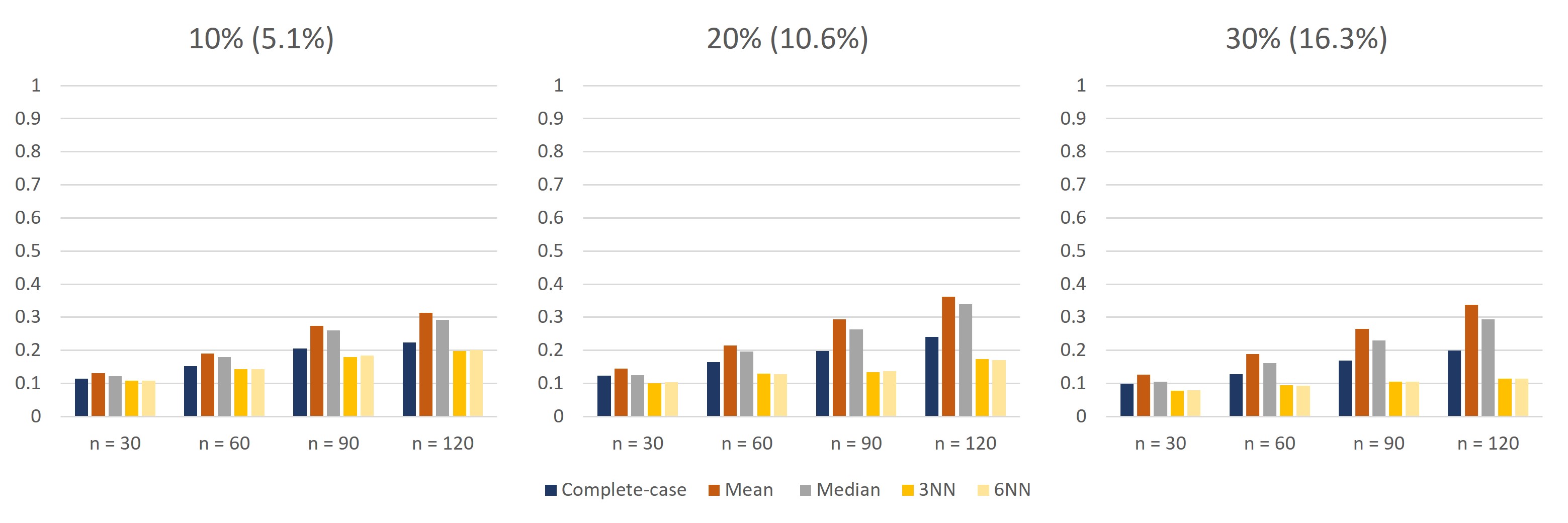}
    \caption{Powers for underlying 2D $t_{11}$ distribution and MCAR data}
    \label{fig:2d_t11}
\end{figure}

\begin{figure}[ht]
    \centering
    \includegraphics[width=\textwidth]{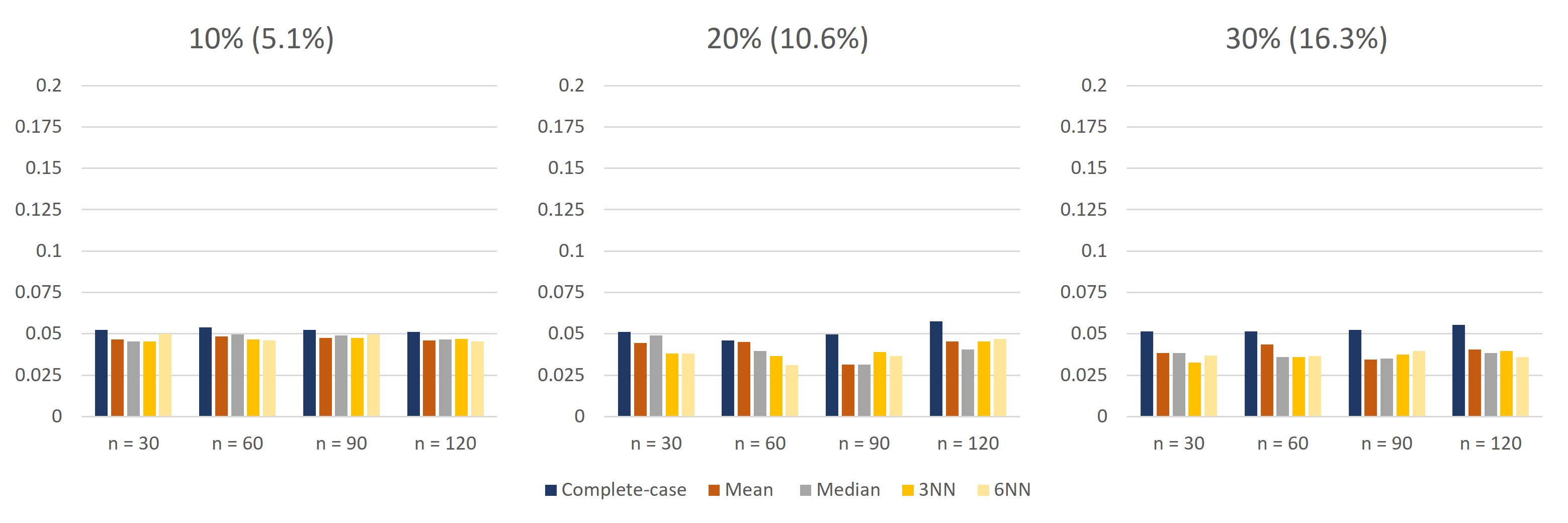}
    \caption{Empirical test sizes for underlying 2D normal distribution with covariance matrix $\Sigma_1$ and MCAR data}
    \label{fig:2d_norm_05}
\end{figure}

\begin{figure}[ht]
    \centering
    \includegraphics[width=\textwidth]{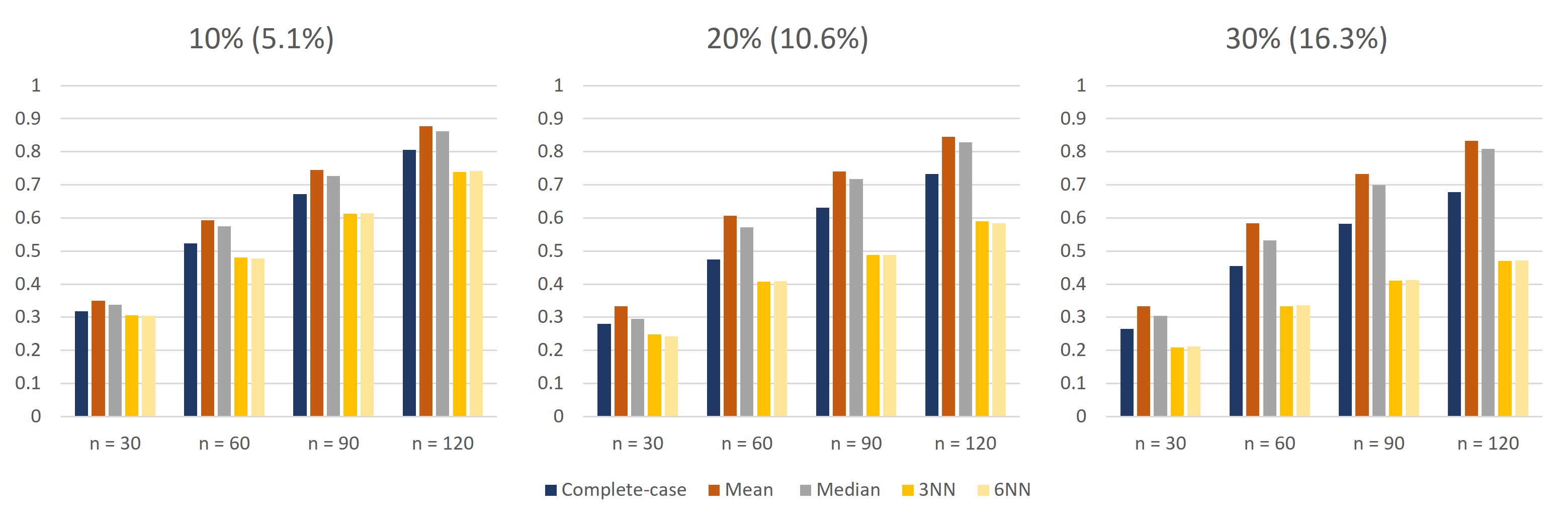}
    \caption{Powers for underlying 2D $t_5$ distribution with scale matrix $\Sigma_1$ and MCAR data}
    \label{fig:2d_t5_05}
\end{figure}

\begin{figure}[ht]
    \centering
    \includegraphics[width=\textwidth]{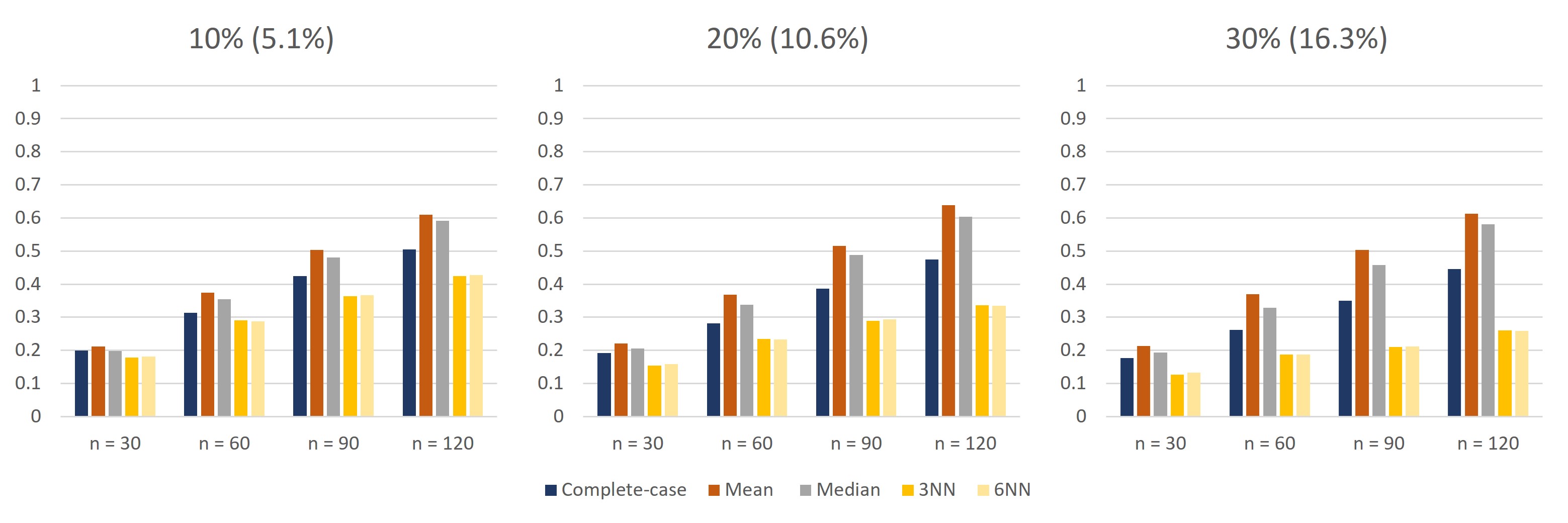}
    \caption{Powers for underlying 2D $t_7$ distribution with scale matrix $\Sigma_1$ and MCAR data}
    \label{fig:2d_t7_05}
\end{figure}

\begin{figure}[ht]
    \centering
    \includegraphics[width=\textwidth]{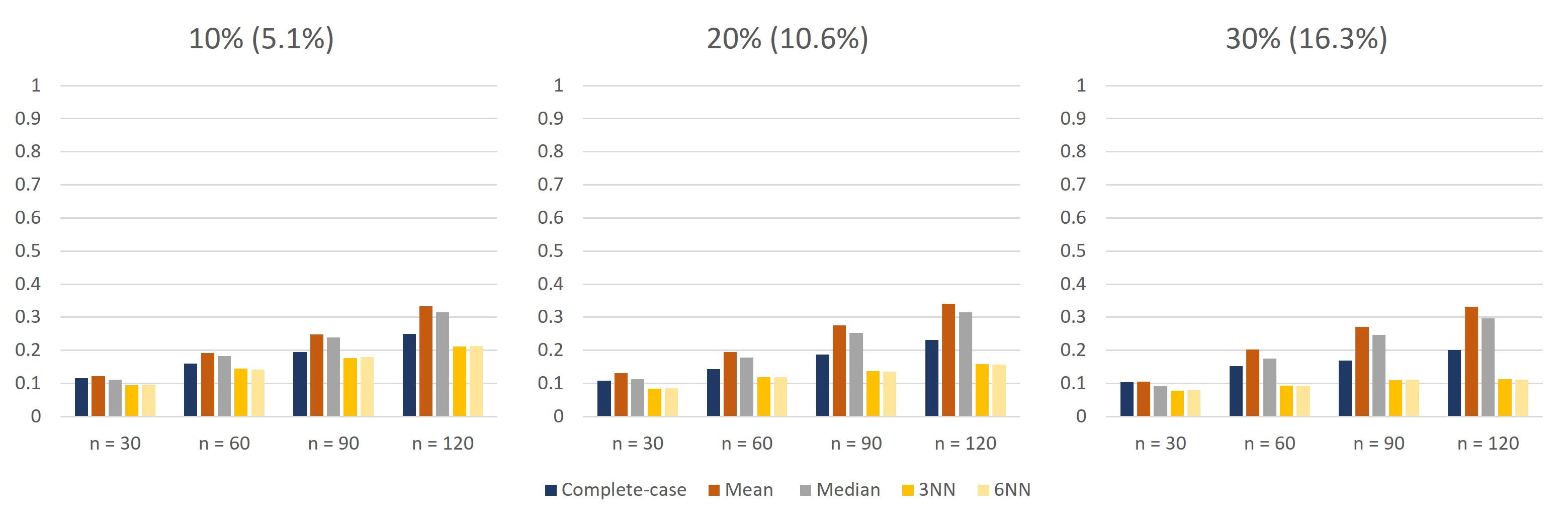}
    \caption{Powers for underlying 2D $t_{11}$ distribution with scale matrix $\Sigma_1$ and MCAR data}
    \label{fig:2d_t11_05}
\end{figure}



\begin{figure}[ht]
    \centering
    \includegraphics[width=\textwidth]{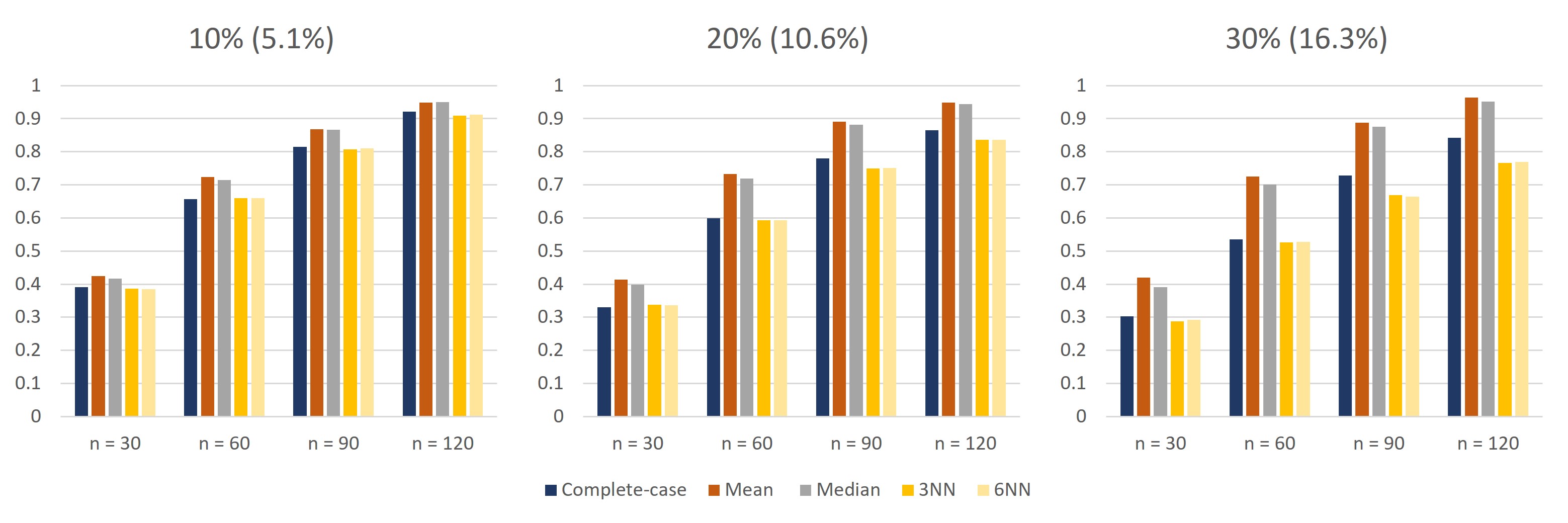}
    \caption{Powers for underlying 3D $t_5$ distribution and MCAR data}
    \label{fig:3d_t5}
\end{figure}

\begin{figure}[ht]
    \centering
    \includegraphics[width=\textwidth]{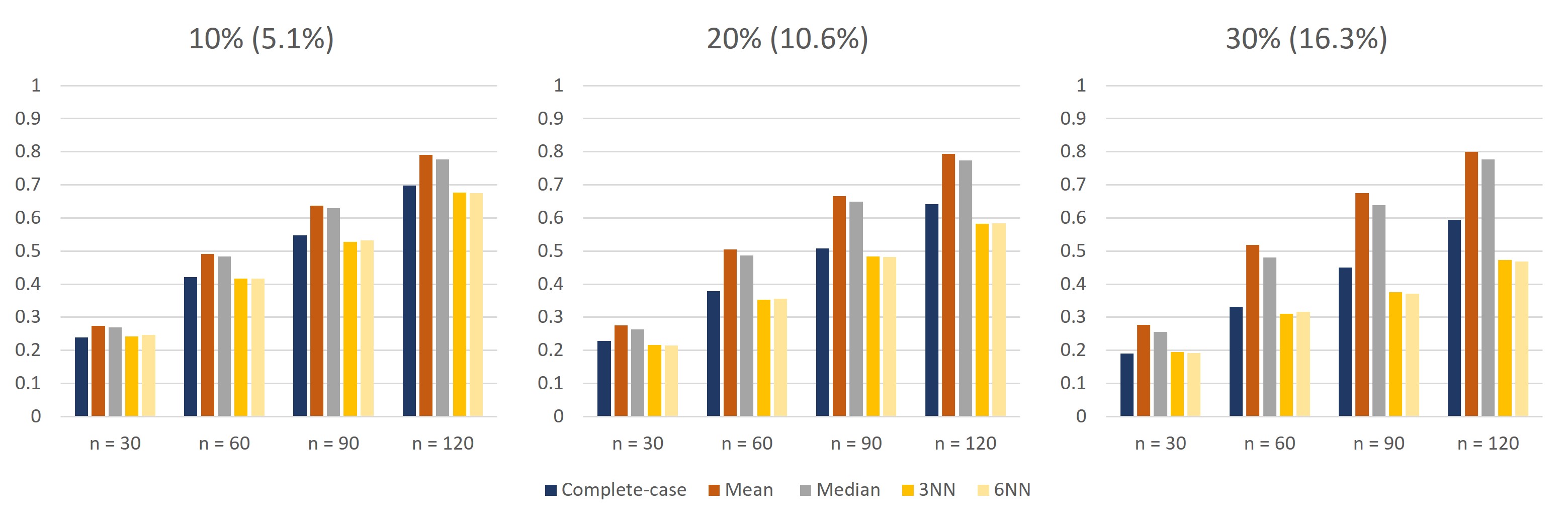}
    \caption{Powers for underlying 3D $t_7$ distribution and MCAR data}
    \label{fig:3d_t7}
\end{figure}

\begin{figure}[ht]
    \centering
    \includegraphics[width=\textwidth]{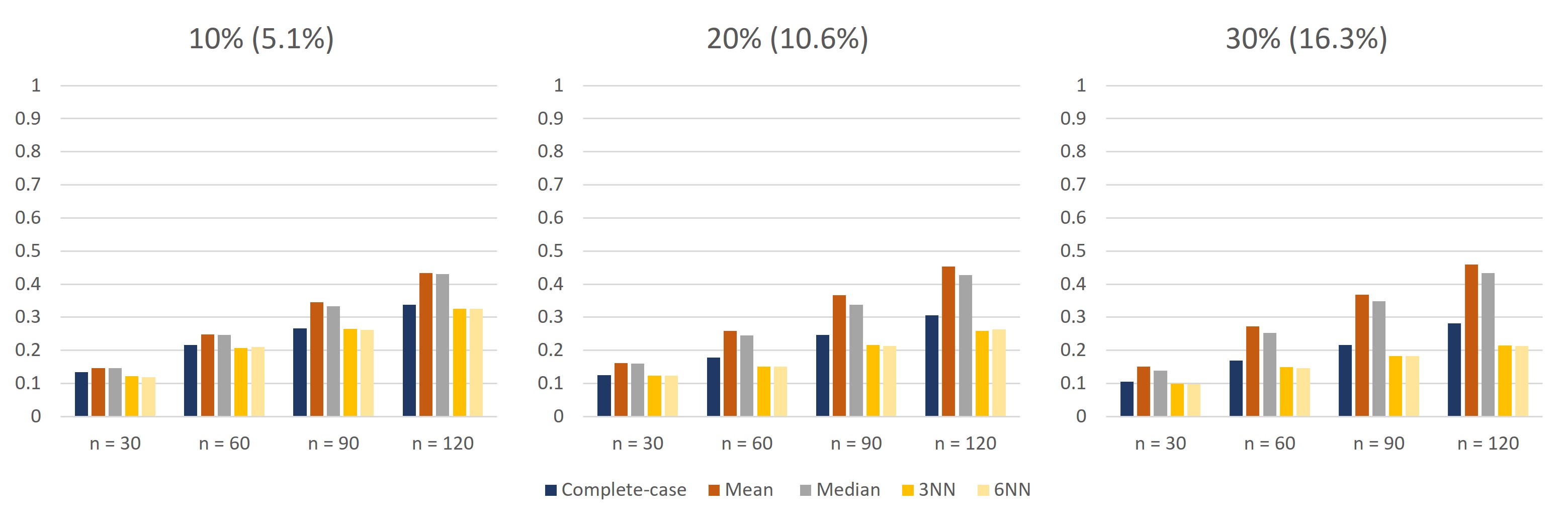}
    \caption{Powers for underlying 3D $t_{11}$ distribution and MCAR data}
    \label{fig:3d_t11}
\end{figure}

\begin{figure}[ht]
    \centering
    \includegraphics[width=\textwidth]{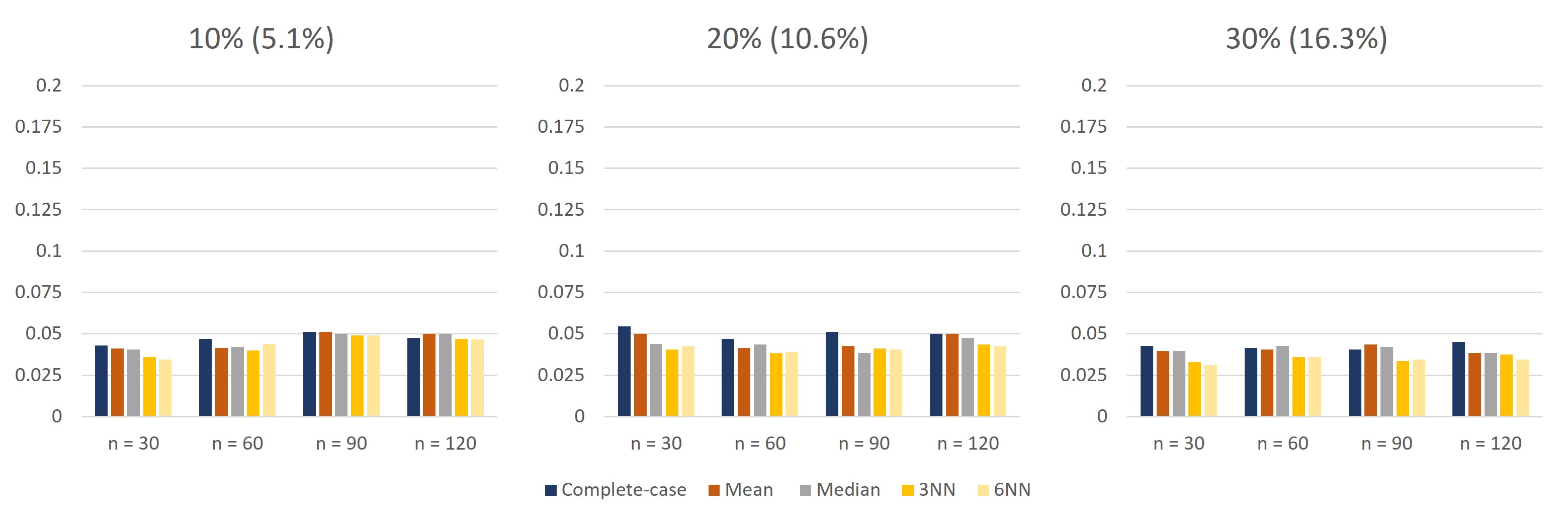}
    \caption{Empirical test sizes for underlying 3D normal distribution with covariance matrix $\Sigma_2$ and MCAR data}
    \label{fig:3d_norm_sgm2}
\end{figure}

\begin{figure}[ht]
    \centering
    \includegraphics[width=\textwidth]{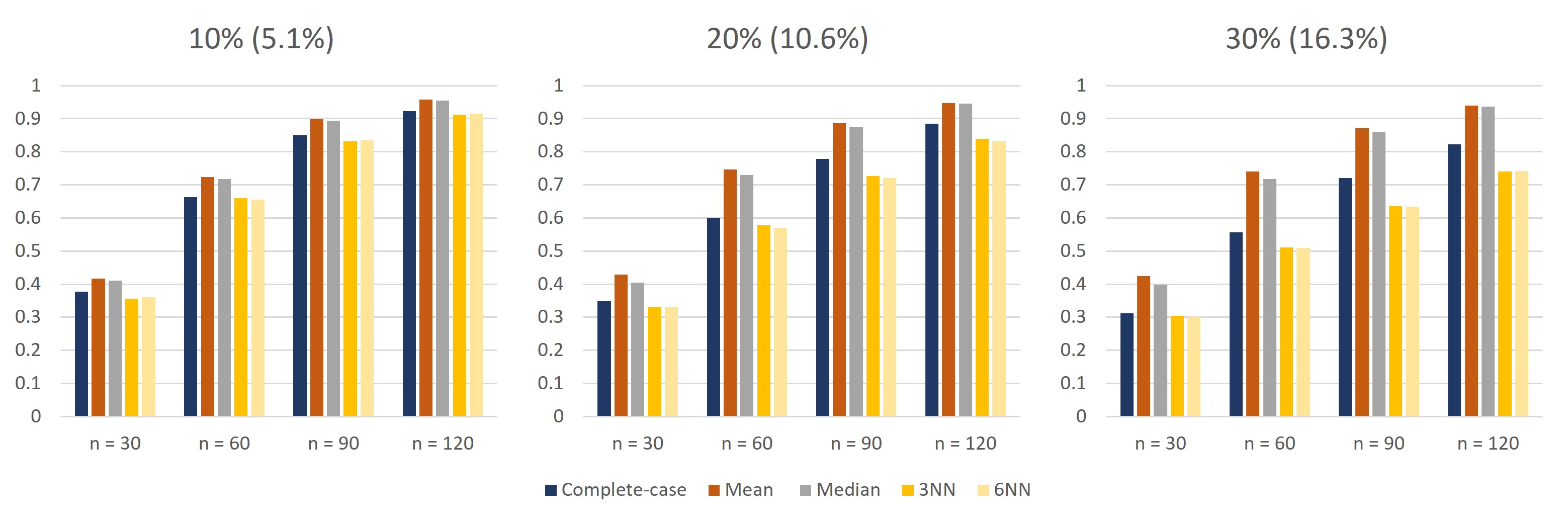}
    \caption{Powers for underlying 3D $t_5$ distribution with scale matrix $\Sigma_2$ and MCAR data}
    \label{fig:3d_t5_sgm2}
\end{figure}

\begin{figure}[ht]
    \centering
    \includegraphics[width=\textwidth]{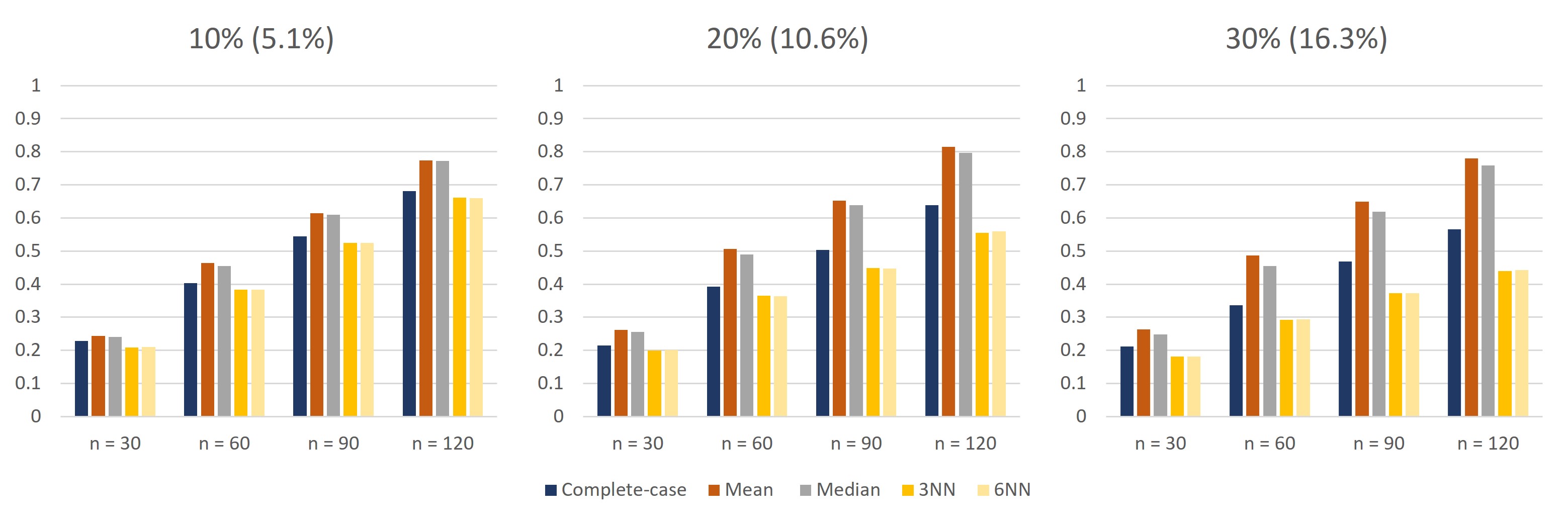}
    \caption{Powers for underlying 3D $t_7$ distribution with scale matrix $\Sigma_2$ and MCAR data}
    \label{fig:3d_t7_sgm2}
\end{figure}

\begin{figure}[ht]
    \centering
    \includegraphics[width=\textwidth]{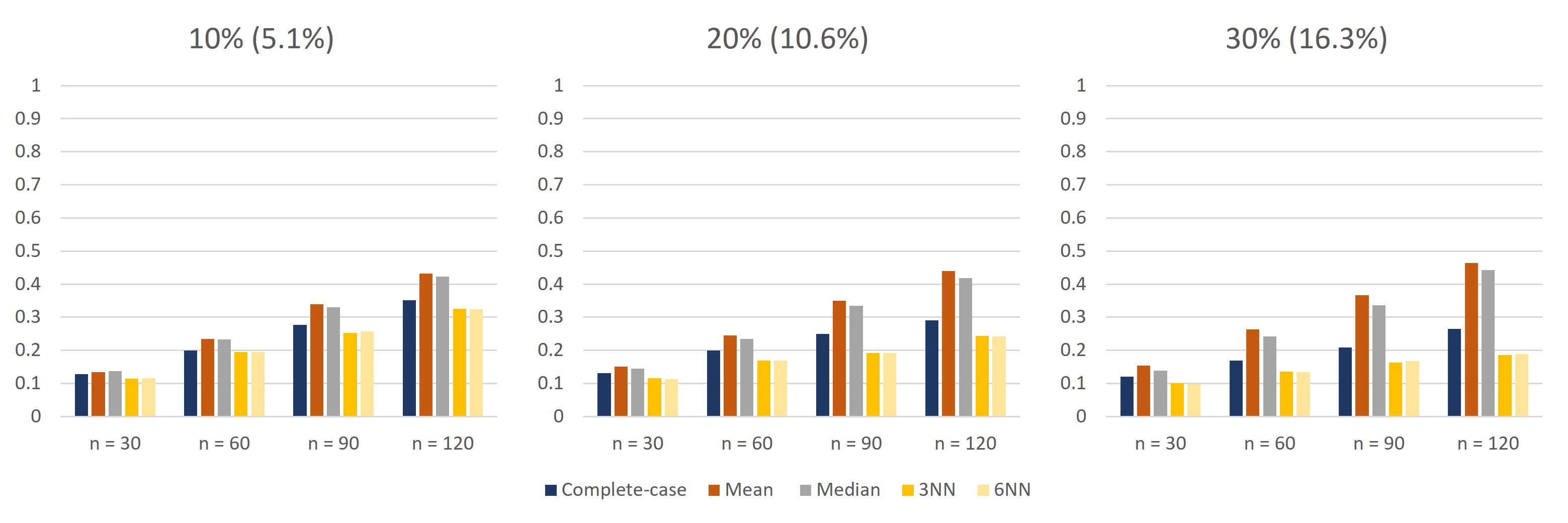}
    \caption{Powers for underlying 3D $t_{11}$ distribution with scale matrix $\Sigma_2$ and MCAR data}
    \label{fig:3d_t11_sgm2}
\end{figure}

\end{document}